\documentclass[twoside]{article}%
\usepackage{amssymb}
\usepackage{amsfonts}
\usepackage{amsmath}
\usepackage{graphicx}%
\providecommand{\U}[1]{\protect\rule{.1in}{.1in}}
\newtheorem{theorem}{Theorem}

\newtheorem{lemma}[theorem]{Lemma}

\newtheorem{remark}[theorem]{Remark}

\newenvironment{proof}[1][Proof]{\noindent\textbf{#1.} }{\ \rule{0.5em}{0.5em}}
\begin{document}

\title{\textbf{Line Solutions for the Euler and Euler-Poisson Equations with Multiple
Gamma Law}}
\author{\textsc{Ling Hei Yeung\thanks{E-mail address: lightisgood2005@yahoo.com.hk}}
\and M\textsc{anwai Yuen\thanks{E-mail address: nevetsyuen@hotmail.com }}\\\textit{Department of Applied Mathematics, The Hong Kong Polytechnic
University,}\\\textit{Hung Hom, Kowloon, Hong Kong}}
\date{Revised 20-May-2010}
\maketitle

\begin{abstract}
In this paper, we study the Euler and Euler-Poisson equations in $R^{N}$, with
multiple $\gamma$-law for pressure function:
\begin{equation}
P(\rho)=e^{s}\sum_{j=1}^{m}\rho^{\gamma_{j}},
\end{equation}
where all $\gamma_{i+1}>\gamma_{i}\geq1$, is the constants. The analytical
line solutions are constructed for the systems. It is novel to discover the
analytical solutions to handle the systems with mixed pressure function. And
our solutions can be extended to the systems with the generalized multiple
damping and pressure function.

Key Words: Multiple Gamma Law, Euler Equations, Euler-Poisson Equations,
Analytical Solutions, Navier-Stokes Equations, Global Solutions, External
forces, Free Boundary, Multiple Damping

\end{abstract}

\section{Introduction}

The $N$-dimensional Euler equations can be formulated as the follows:%
\begin{equation}
\left\{
\begin{array}
[c]{c}%
{\normalsize \rho}_{t}{\normalsize +\nabla\cdot(\rho\vec{u})=}{\normalsize 0,}%
\\
\rho\left[  \vec{u}_{t}+\left(  \vec{u}\cdot\nabla\right)  \vec{u}\right]
+\nabla P(\rho){\normalsize =-}\delta\rho\bigtriangledown\Phi-\rho\vec
{F}(t),\\
S_{t}+{\normalsize \vec{u}\cdot\nabla S}=0,\\
{\normalsize \Delta\Phi(t,\vec{x})=\alpha(N)}{\normalsize \rho,}%
\end{array}
\right.  \label{Euler-Poisson11}%
\end{equation}
with $\vec{x}=(x_{1},x_{2},...,x_{N})\in R^{N},$\newline and $\rho=\rho
(t,\vec{x}),$ $\vec{u}(t,\vec{x})\in R^{N}$ and ${\normalsize S}=S(t)$ are the
density, the velocity and the entropy respectively. And $\alpha(N)$ is a
constant related to the unit ball in $R^{N}$: $\alpha(1)=2$; $\alpha(2)=2\pi$
and for $N\geq3,$%
\begin{equation}
\alpha(N)=N(N-2)Vol(N)=N(N-2)\frac{\pi^{N/2}}{\Gamma(N/2+1)},
\end{equation}
where $Vol(N)$ is the volume of the unit ball in $R^{N}$ and $\Gamma$ is a
Gamma function.

When $\delta=1$, the system can model fluids that are self-gravitating , such
as gaseous stars. For $N=3$, the equations (\ref{Euler-Poisson11}) are the
classical (non-relativistic) descriptions of a galaxy in astrophysics. See
\cite{BT} and \cite{C}, for details about the systems.\newline When
$\delta=-1$, the system is the compressible Euler-Poisson equations with
repulsive forces. The equation (\ref{Euler-Poisson11})$_{4}$ is the Poisson
equation through which the potential with repulsive forces is determined by
the density distribution of the electrons. In this case, the system can be
viewed as a semiconductor model. See \cite{Cse} and \cite{Lions} for detailed
analysis of the system.\newline When $\delta=0$, the potential forces are
ignored. The system is called the Euler equations. See \cite{CW} and \cite{Ni}
for detailed analysis of the system.

Here $P=P(\rho)$ is the pressure, the $\gamma$-law on the pressure for the
single gas, i.e.
\begin{equation}
P(\rho)=e^{s}\rho^{\gamma}, \label{eq2}%
\end{equation}
is a universal hypothesis. The constant $\gamma=c_{P}/c_{v}\geq1$, where
$c_{p}$ and $c_{v}$ are the specific heats per unit mass under constant
pressure and constant volume respectively, is the ratio of the specific heats.
In particular, the fluid is called isothermal if $\gamma=1$. More generally,
the pressure function of mixed gases, can be expressed by the multiple
$\gamma$-law (\cite{LRTZ}, \cite{PV}, and \cite{TS}):%
\begin{equation}
P(\rho)=P_{1}+P_{2}+...P_{N}=e^{s}\sum_{j=1}^{m}\rho^{\gamma_{j}},
\end{equation}
where all $\gamma_{j+1}>\gamma_{j}\geq1$, is the constant, $m$ is a positive
integer.\newline The system with the multiple $\gamma$-law can reflect the
better approximation of the real situations. For example, the fluids in the
stars mix with many types of gases in models of gaseous stars in astrophysics
\cite{TS}.

Here the time-dependent external force $\vec{F}(t)=(F_{1}(t),F_{2}%
(t),...,F_{N}(t))\in C^{0}$ are coupled in the systems.

The system (\ref{Euler-Poisson11}) can be rewritten in scalar form,%
\begin{equation}
\left\{
\begin{array}
[c]{c}%
\frac{\partial\rho}{\partial t}+\sum_{k=1}^{N}u_{k}\frac{\partial\rho
}{\partial x_{k}}+\rho\sum_{k=1}^{N}\frac{\partial u_{k}}{\partial x_{k}%
}{\normalsize =}{\normalsize 0,}\\
\rho\left(  \frac{\partial u_{i}}{\partial t}+\sum_{k=1}^{N}u_{k}%
\frac{\partial u_{i}}{\partial x_{k}}\right)  +\frac{\partial}{\partial x_{i}%
}\left(  e^{S}\sum_{j=1}^{m}\rho^{\gamma_{j}}\right)  {\normalsize =-\rho
}\frac{\delta\partial}{\partial x_{i}}\Phi+\rho F_{i}(t)\text{, for
}i=1,2,...N,\\
S_{t}+\sum_{k=1}^{N}u_{k}\frac{\partial S}{\partial x_{k}}=0.
\end{array}
\right.  \label{scalar form}%
\end{equation}

For the Euler equations (\ref{scalar form}), $(\delta=0),$ with $\vec
{F}(t)=0,$ in radial symmetry:
\begin{equation}
\rho(t,\vec{x})=\rho(t,r)\text{ and }\vec{u}=\frac{\vec{x}}{r}V(t,r):=\frac
{\vec{x}}{r}V,
\end{equation}
with $r=\left(  \sum_{i=1}^{N}x_{i}^{2}\right)  ^{1/2}$,\newline there exists
a family of solutions for the Euler equations (\ref{scalar form}),
$(\delta=0),$ without the external force $(\vec{F}=0)$, for $\gamma>1$,
\cite{Li}%
\begin{equation}
\left\{
\begin{array}
[c]{c}%
\rho(t,r)=\left\{
\begin{array}
[c]{cc}%
\frac{^{y(r/a(t))^{1/(\gamma-1)}}}{a(t)^{N}}, & \text{ for }y(\frac{r}%
{a(t)})\geq0;\\
0, & \text{for }y(\frac{r}{a(t)})<0
\end{array}
\right.  ,\text{ }V(t,r)=\frac{\overset{\cdot}{a}(t)}{a(t)}r,\text{
}S(t,r)=\ln K,\\
\overset{\cdot\cdot}{a}(t)=\frac{-\lambda}{a(t)^{^{1+N(\gamma-1)}}},\text{
}a(0)=a_{0}>0,\text{ }\overset{\cdot}{a}(0)=a_{1},\\
y(x)=\frac{(\gamma-1)\lambda}{2\gamma K}x^{2}+\alpha^{\theta-1},
\end{array}
\right.  ;
\end{equation}
where $K$ is a positive number,\newline for $\gamma=1,$ \cite{Y2}%
\begin{equation}
\left\{
\begin{array}
[c]{c}%
\rho(t,r)=\frac{e^{y(r/a(t))}}{a(t)^{N}},V(t,r)=\frac{\overset{\cdot}{a}%
(t)}{a(t)}r,\text{ }S(t,r)=\ln K\\
\overset{\cdot\cdot}{a}(t)=\frac{-\lambda}{a(t)},a(0)=a_{0}>0,\overset{\cdot
}{a}(0)=a_{1},\\
y(x)=\frac{\lambda}{2K}x^{2}+\alpha,
\end{array}
\right.
\end{equation}
where $\lambda,$ $\alpha,$ $a_{0}$ and $a_{1}$ are constants.

The separation method for analytical solutions, were used to handle other
similar systems with single $\gamma$ functions:%
\begin{equation}
P(\rho)=e^{s}\rho^{\gamma},
\end{equation}
or without pressure function, in \cite{GW}, \cite{DXY}, \cite{Li}, \cite{M1},
\cite{Y}, \cite{Y1}, \cite{Y2}, \cite{Y4} and \cite{Y5}.

It is very natural to extend the results for the system with multiple $\gamma$
function (\ref{scalar form}):%
\begin{equation}
P(\rho)=e^{s}\sum_{j=1}^{m}\rho^{\gamma_{j}}.
\end{equation}
In this article, we have obtained a class of line solutions to the Euler
equations (\ref{scalar form}) $(\delta=0)$ and Euler-Poisson equations
(\ref{scalar form}) $(\delta=\pm1)$, with multiple Gamma function, in the
following theorems:

\begin{theorem}
\label{thm:1}For the Euler equations with multiple Gamma function,
$(\delta=0)$ (\ref{scalar form}), we have the family of the solutions,%
\begin{equation}
\left\{
\begin{array}
[c]{c}%
\rho(t,\vec{x})=f\left(  \overset{N}{\underset{i=1}{\Sigma}}C_{i}\left(
x_{i}-a_{i}(t)\right)  \right)  ,\\
\vec{u}(t,\vec{x})=(\dot{a}_{1}(t),\text{ }\dot{a}_{2}(t),...,\text{ }\dot
{a}_{N}(t)),\\
S(t,\vec{x})=\ln\left[  g\left(  \overset{N}{\underset{i=1}{\Sigma}}%
C_{i}\left(  x_{i}-a_{i}(t)\right)  \right)  \right]  ,
\end{array}
\right.  \label{yy1}%
\end{equation}
where%
\begin{equation}
\left\{
\begin{array}
[c]{c}%
\ddot{a}_{i}(t)=F_{i}(t)+C_{i}\xi,\text{ for }i=1,2,...N,\\
\text{ }a(0)=a_{0},\text{ }\dot{a}(0)=a_{1},
\end{array}
\right.
\end{equation}
(1)for $\gamma_{1}>1,$ with%
\begin{equation}
\left\{
\begin{array}
[c]{c}%
\xi+\dot{g}(z)\sum_{j=1}^{m}f\left(  z\right)  ^{\gamma_{j}-1}+g(z)\sum
_{j=1}^{m}\gamma_{j}f\left(  z\right)  ^{\gamma_{i}-2}\dot{f}\left(  z\right)
=0,\\
g(z)>0,\text{ for }z\in(-\infty,\infty),
\end{array}
\right.  ;
\end{equation}
where $C_{1},$ $C_{2},...,$ $C_{N},$ $\xi,$ $g_{0}$ and $g_{1}$ are arbitrary
constants; and $f\geq0$ is an arbitrary $C^{1}$ function;\newline(2)for
$\gamma_{1}=1$, with%
\begin{equation}
\left\{
\begin{array}
[c]{c}%
\xi+\rho C_{i}\dot{g}(z)+g(z)\frac{\dot{f}\left(  z\right)  }{f(z)}C_{i}%
+\dot{g}(z)\sum_{j=2}^{m}f\left(  z\right)  ^{\gamma_{j}-1}+g(z)\sum_{j=2}%
^{m}\gamma_{j}f\left(  z\right)  ^{\gamma_{i}-2}\dot{f}\left(  z\right)  =0,\\
g(z)>0,\text{ for }z\in(-\infty,\infty),
\end{array}
\right.
\end{equation}
$f>0$ is an arbitrary $C^{1}$ function.
\end{theorem}

We notice that the velocity $\vec{u}(t,\vec{x})$ of the solutions (\ref{yy1})
and (\ref{yy2}), are only time-dependent functions:%
\begin{equation}
\vec{u}(t,\vec{x})=(\dot{a}_{1}(t),\dot{a}_{2}(t),...,\dot{a}_{N}(t)).
\end{equation}
It is different from the conventional velocity for the analytical solutions:%
\begin{equation}
\vec{u}(t,\vec{x})=\frac{\overset{\cdot}{a}(t)}{a(t)}\vec{x}.
\end{equation}

Moreover, we need some modification to have the corresponding results for the
Euler-Poisson equations.

\begin{theorem}
\label{thm:2}For the Euler-Poisson equations with multiple Gamma function,
$(\delta=\pm1),$ (\ref{scalar form}), we have the family of the solutions,%
\begin{equation}
\left\{
\begin{array}
[c]{c}%
\rho(t,\vec{x})=\frac{\overset{N}{\underset{i=1}{\Sigma}}C_{i}^{2}}{\alpha
(N)}\ddot{f}\left(  \overset{N}{\underset{i=1}{\Sigma}}C_{i}\left(
x_{i}-a_{i}(t)\right)  \right)  ,\\
\vec{u}(t,\vec{x})=(\dot{a}_{1}(t),\text{ }\dot{a}_{2}(t),...,\text{ }\dot
{a}_{N}(t)),\\
S(t,\vec{x})=\ln\left[  g\left(  \overset{N}{\underset{i=1}{\Sigma}}%
C_{i}\left(  x_{i}-a_{i}(t)\right)  \right)  \right]
\end{array}
\right.  \label{yy2}%
\end{equation}
where%
\begin{equation}
\left\{
\begin{array}
[c]{c}%
\ddot{a}_{i}(t)=F_{i}(t)+C_{i}\xi+d_{i}(t),\text{ for }i=1,2,...N,\\
\text{ }a(0)=a_{0},\text{ }\dot{a}(0)=a_{1},
\end{array}
\right.  \label{at}%
\end{equation}
(1)for $\gamma_{1}>1$, with%
\begin{equation}
\left\{
\begin{array}
[c]{c}%
\xi+\frac{\alpha(N)}{\overset{N}{\underset{i=1}{\Sigma}}C_{i}^{2}}\dot
{g}(z)\sum_{j=1}^{m}f\left(  z\right)  ^{\gamma_{j}-1}+\frac{\alpha
(N)}{\overset{N}{\underset{i=1}{\Sigma}}C_{i}^{2}}g(z)\sum_{j=1}^{m}\gamma
_{j}f\left(  z\right)  ^{\gamma_{i}-2}\dot{f}\left(  z\right)  +\delta\dot
{f}(z)=0,\\
g(z)>0,\text{ for }z\in(-\infty,\infty),
\end{array}
\right.  \label{entropyODE}%
\end{equation}
where $C_{1},$ $C_{2},...,$ $C_{N}$, $\xi,$ $g_{0}$ and $g_{1}$ are arbitrary
constants with $\overset{N}{\underset{i=1}{\Sigma}}C_{i}^{2}>0$; $\ddot{f}%
\geq0$ is an arbitrary $C^{3}$ function; and $d_{1},d_{2},....,d_{N}$ are
arbitrary $C^{0}$ functions;\newline(2)for $\gamma_{1}=1$, with%
\begin{equation}
\left\{
\begin{array}
[c]{c}%
\xi+\frac{\alpha(N)}{\overset{N}{\underset{i=1}{\Sigma}}C_{i}^{2}}\dot
{g}(z)+\frac{\alpha(N)}{\overset{N}{\underset{i=1}{\Sigma}}C_{i}^{2}}%
g(z)\frac{\dddot{f}\left(  z\right)  }{f(z)}\\
+\frac{\alpha(N)}{\overset{N}{\underset{i=1}{\Sigma}}C_{i}^{2}}\dot{g}%
(z)\sum_{j=2}^{m}\ddot{f}\left(  z\right)  ^{\gamma_{j}-1}+\frac{\alpha
(N)}{\overset{N}{\underset{i=1}{\Sigma}}C_{i}^{2}}g(z)\sum_{j=2}^{m}\gamma
_{j}\ddot{f}\left(  z\right)  ^{\gamma_{j}-2}\dddot{f}\left(  z\right)
+\delta\dot{f}(z)=0,\\
g(z)>0,\text{ for }z\in(-\infty,\infty),
\end{array}
\right.
\end{equation}
$\dddot{f}>0$ is an arbitrary $C^{1}$ function.
\end{theorem}

\begin{remark}
The mass of the solutions (\ref{yy1}) for the Euler equations, and (\ref{yy2})
for the Euler-Poisson equations, in 1-dimensional case, is finite, if%
\begin{equation}
\int_{-\infty}^{\infty}f(z)dz<\infty,
\end{equation}
and
\begin{equation}
\int_{-\infty}^{\infty}\ddot{f}(z)dz<\infty,
\end{equation}
respectively.
\end{remark}

\section{Line Solutions}

With regard to the continuity equation of mass (\ref{Euler-Poisson11})$_{1}$,
we found the following solution structures of the below lemmas fit it well:

\begin{lemma}
For the mass equation:
\begin{equation}
{\normalsize \rho}_{t}{\normalsize +\nabla\cdot(\rho\vec{u})=}{\normalsize 0,}%
\end{equation}
there exist solutions,%
\begin{equation}
\rho(t,\vec{x})=f\left(  \overset{N}{\underset{i=1}{\Sigma}}C_{i}\left(
x_{i}-a_{i}(t)\right)  \right)  ,\text{ }{\normalsize u(t,\vec{x})=}\vec
{u}=(\dot{a}_{1}(t),\text{ }\dot{a}_{2}(t),...,\text{ }\dot{a}_{N}(t)),
\end{equation}
with the form arbitrary $f\geq0\in C^{1}$ and arbitrary $a_{i}(t)\in C^{1}.$

\begin{proof}
For the mass equation, we have
\begin{align}
&  \rho_{t}+\nabla\rho\cdot\vec{u}+\rho\nabla\cdot\vec{u}\\[0.1in]
&  =\frac{\partial}{\partial t}f\left(  \overset{N}{\underset{i=1}{\Sigma}%
}C_{i}\left(  x_{i}-a_{i}(t)\right)  \right)  +\nabla f\left(  \overset
{N}{\underset{i=1}{\Sigma}}C_{i}\left(  x_{i}-a_{i}(t)\right)  \right)
\cdot(\dot{a}_{1}(t),\text{ }\dot{a}_{2}(t),...,\text{ }\dot{a}_{N}%
(t))\\[0.1in]
&  =\dot{f}\left(  \overset{N}{\underset{i=1}{\Sigma}}C_{i}\left(  x_{i}%
-a_{i}(t)\right)  \right)  \left(  -\overset{N}{\underset{i=1}{\Sigma}}%
C_{i}\dot{a}_{i}(t)\right) \\[0.1in]
&  +\dot{f}\left(  \overset{N}{\underset{i=1}{\Sigma}}C_{i}\left(  x_{i}%
-a_{i}(t)\right)  \right)  \left(  C_{1}\dot{a}_{1}(t)+C_{2}\dot{a}%
_{2}(t)+....+C_{N}\dot{a}_{N}(t)\right) \\
&  =\dot{f}\left(  \overset{N}{\underset{i=1}{\Sigma}}C_{i}\left(  x_{i}%
-a_{i}(t)\right)  \right)  \left(  -\overset{N}{\underset{i=1}{\Sigma}}%
C_{i}\dot{a}_{i}(t)+\overset{N}{\underset{i=1}{\Sigma}}C_{i}\dot{a}%
_{i}(t)\right) \\
&  =0.
\end{align}
The proof is completed.
\end{proof}
\end{lemma}

Similarly, we have the corresponding lemma for the entropy equation
(\ref{Euler-Poisson11})$_{3}:$

\begin{lemma}
For the entropy equation:
\begin{equation}
S_{t}+{\normalsize \vec{u}\cdot\nabla S}=0,
\end{equation}
there exist solutions,%
\begin{equation}
S(t,\vec{x})=G\left(  \overset{N}{\underset{i=1}{\Sigma}}C_{i}\left(
x_{i}-a_{i}(t)\right)  \right)  ,\text{ }{\normalsize u(t,\vec{x})=}\vec
{u}=(\dot{a}_{1}(t),\text{ }\dot{a}_{2}(t),...,\text{ }\dot{a}_{N}(t)),
\end{equation}
with the form $f\geq0\in C^{1}$ and $a_{i}(t)\in C^{1}.$

\begin{proof}
For the mass equation, we can obtain:
\begin{align}
&  S_{t}+\vec{u}\cdot\nabla S\\[0.1in]
&  =\frac{\partial}{\partial t}G\left(  \overset{N}{\underset{i=1}{\Sigma}%
}C_{i}\left(  x_{i}-a_{i}(t)\right)  \right)  +(\dot{a}_{1}(t),\text{ }\dot
{a}_{2}(t),...,\text{ }\dot{a}_{N}(t))\cdot\nabla\dot{G}\left(  \overset
{N}{\underset{i=1}{\Sigma}}C_{i}\left(  x_{i}-a_{i}(t)\right)  \right)
\\[0.1in]
&  =\dot{G}\left(  \overset{N}{\underset{i=1}{\Sigma}}C_{i}\left(  x_{i}%
-a_{i}(t)\right)  \right)  \left(  -\overset{N}{\underset{i=1}{\Sigma}}%
C_{i}\dot{a}_{i}(t)\right) \\[0.1in]
&  +\dot{G}\left(  \overset{N}{\underset{i=1}{\Sigma}}C_{i}\left(  x_{i}%
-a_{i}(t)\right)  \right)  \left(  C_{1}\dot{a}_{1}(t)+C_{2}\dot{a}%
_{2}(t)+....+C_{N}\dot{a}_{N}(t)\right) \\
&  =0.
\end{align}
The proof is completed.
\end{proof}
\end{lemma}

The technique of constructing solutions is to deduce the partial differential
equations into ordinary differential equations only. Based on the above
lemmas, it is clear to check our solutions for the system.

\begin{proof}
[Proof of Theorem \ref{thm:1}]Our structure of the solutions (\ref{yy1}), fits
well for the mass equation (\ref{Euler-Poisson11})$_{1}$ and the entropy
equation (\ref{Euler-Poisson11})$_{3}$, from the above lemmas. For the $i$-th
momentum equation (\ref{scalar form})$_{2}$, we can define $z:=\overset
{N}{\underset{i=1}{\Sigma}}C_{i}\left(  x_{i}-a_{i}(t)\right)  ,$ to have
\begin{align}
&  \rho\left(  \frac{\partial u_{i}}{\partial t}+\sum_{k=1}^{N}u_{k}%
\frac{\partial u_{i}}{\partial x_{k}}\right)  +\frac{\partial}{\partial x_{i}%
}e^{S}\sum_{j=1}^{m}\rho^{\gamma_{j}}+\rho F_{i}(t)\\[0.1in]
&  =\rho\ddot{a}(t)+\frac{\partial}{\partial x_{i}}e^{\ln g(z)}\sum_{j=1}%
^{m}f\left(  z\right)  ^{\gamma_{j}}+\rho F_{i}(t)\\[0.1in]
&  =\rho\ddot{a}(t)+\left(  \frac{\partial}{\partial x_{i}}g\left(  z\right)
\right)  C_{i}\sum_{j=1}^{m}f\left(  z\right)  ^{\gamma_{j}}+g(z)\frac
{\partial}{\partial x_{i}}\left(  \sum_{j=1}^{m}f\left(  z\right)
^{\gamma_{j}}\right)  +\rho F_{i}(t)\\[0.1in]
&  =\rho\ddot{a}(t)+\rho C_{i}\dot{g}(z)\sum_{j=1}^{m}f\left(  z\right)
^{\gamma_{j}-1}+g(z)\sum_{j=1}^{m}\gamma_{j}f\left(  z\right)  ^{\gamma_{j}%
-1}\dot{f}\left(  z\right)  C_{i}+\rho F_{i}(t)\\[0.1in]
&  =\rho C_{i}\left\{  \xi+\dot{g}(z)\sum_{j=1}^{m}f\left(  z\right)
^{\gamma_{j}-1}+g(z)\sum_{j=1}^{m}\gamma_{j}f\left(  z\right)  ^{\gamma_{i}%
-2}\dot{f}\left(  z\right)  \right\} \\[0.1in]
&  =0,
\end{align}
where we require the following ordinary differential equations:%
\begin{equation}
\left\{
\begin{array}
[c]{c}%
\ddot{a}_{i}(t)=F_{i}(t)+C_{i}\xi,\text{ for }i=1,2,...N,\\
a(0)=a_{0},\text{ }\dot{a}(0)=a_{1},\\
\xi+\dot{g}(z)\sum_{j=1}^{m}f\left(  z\right)  ^{\gamma_{j}-1}+g(z)\sum
_{j=1}^{m}\gamma_{j}f\left(  z\right)  ^{\gamma_{i}-2}\dot{f}\left(  z\right)
=0,\\
g(z)>0,\text{ for }z\in(-\infty,\infty).
\end{array}
\right.
\end{equation}

For $\gamma_{1}=1,$ we have,%
\begin{align}
&  \rho\left(  \frac{\partial u_{i}}{\partial t}+\sum_{k=1}^{N}u_{k}%
\frac{\partial u_{i}}{\partial x_{k}}\right)  +\frac{\partial}{\partial x_{i}%
}\left(  e^{S}\rho+e^{S}\sum_{j=2}^{m}\rho^{\gamma_{j}}\right)  +\rho
F_{i}(t)\\[0.1in]
&  =\rho\ddot{a}(t)+\frac{\partial}{\partial x_{i}}\left(  e^{\ln g(z)}%
\rho\right)  +\frac{\partial}{\partial x_{i}}\left(  e^{\ln g(z)}\sum
_{j=2}^{m}f\left(  z\right)  ^{\gamma_{j}}\right)  +\rho F_{i}(t)\\[0.1in]
&  =\rho\ddot{a}(t)+\left(  \frac{\partial}{\partial x_{i}}g\left(  z\right)
\right)  f(z)+g(z)\frac{\partial}{\partial x_{i}}f(z)\\[0.1in]
&  +\left(  \frac{\partial}{\partial x_{i}}g\left(  z\right)  \right)
C_{i}\sum_{j=2}^{m}f\left(  z\right)  ^{\gamma_{j}}+g(z)\frac{\partial
}{\partial x_{i}}\left(  \sum_{j=2}^{m}f\left(  z\right)  ^{\gamma_{j}%
}\right)  +\rho F_{i}(t)\\
&  =\rho\ddot{a}(t)+\rho C_{i}\dot{g}(z)f\left(  z\right)  +g(z)\dot{f}\left(
z\right)  C_{i}\\[0.1in]
&  +\rho C_{i}\dot{g}(z)\sum_{j=2}^{m}f\left(  z\right)  ^{\gamma_{j}%
-1}+g(z)\sum_{j=2}^{m}\gamma_{j}f\left(  z\right)  ^{\gamma_{j}-1}\dot
{f}\left(  z\right)  C_{i}+\rho F_{i}(t)\\
&  =\rho C_{i}\left\{  \xi+\rho C_{i}\dot{g}(z)+g(z)\frac{\dot{f}\left(
z\right)  }{f(z)}C_{i}+\dot{g}(z)\sum_{j=2}^{m}f\left(  z\right)  ^{\gamma
_{j}-1}+g(z)\sum_{j=2}^{m}\gamma_{j}f\left(  z\right)  ^{\gamma_{i}-2}\dot
{f}\left(  z\right)  \right\} \\[0.1in]
&  =0,
\end{align}
with
\begin{equation}
\left\{
\begin{array}
[c]{c}%
\xi+\rho C_{i}\dot{g}(z)+g(z)\frac{\dot{f}\left(  z\right)  }{f(z)}C_{i}%
+\dot{g}(z)\sum_{j=2}^{m}f\left(  z\right)  ^{\gamma_{j}-1}+g(z)\sum_{j=2}%
^{m}\gamma_{j}f\left(  z\right)  ^{\gamma_{i}-2}\dot{f}\left(  z\right)  =0,\\
g(z)>0,\text{ for }z\in(-\infty,\infty).
\end{array}
\right.
\end{equation}
The proof is completed.
\end{proof}

With similar analysis, we can derive the corresponding theorem for the
Euler-Poisson equations.

\begin{proof}
[Proof of Theorem \ref{thm:2}]We can use the previous lemmas again to handle
the mass equation and entropy equation again. Then, we can definite the
potential function as:
\begin{equation}
\Phi(t,\vec{x})=f(z)+\overset{N}{\sum_{i=1}}d_{i}(t)x_{i}, \label{potential}%
\end{equation}
where $z:=\overset{N}{\underset{i=1}{\Sigma}}C_{i}\left(  x_{i}-a_{i}%
(t)\right)  .$

We differentiate (\ref{potential}) twice to obtain:
\begin{equation}
\Delta\Phi(t,\vec{x})=\ddot{f}(z)\overset{N}{\underset{i=1}{\Sigma}}C_{i}%
^{2}=\alpha(N)\rho,
\end{equation}
where in our solution%
\begin{equation}
\rho=\frac{\overset{N}{\underset{i=1}{\Sigma}}C_{i}^{2}}{\alpha(N)}\ddot
{f}(z).
\end{equation}
For the $i$-th momentum equation (\ref{scalar form})$_{2}$, we have:%
\begin{align}
&  \rho\left(  \frac{\partial u_{i}}{\partial t}+\sum_{k=1}^{N}u_{k}%
\frac{\partial u_{i}}{\partial x_{k}}\right)  +\frac{\partial}{\partial x_{i}%
}\left(  e^{S}\sum_{j=1}^{m}\rho^{\gamma_{j}}\right)  +\delta\rho
\frac{\partial}{\partial x_{i}}\Phi+\rho F_{i}(t)\\
&  =\rho\dot{a}_{it}(t)+\frac{\partial}{\partial x_{i}}\left(  e^{\ln g\left(
z\right)  }\sum_{j=1}^{m}\ddot{f}\left(  z\right)  ^{\gamma_{j}}\right)
+\delta\rho\frac{\partial}{\partial x_{i}}\left(  f(z)+\overset{N}{\sum_{i=1}%
}d_{i}(t)x_{i}\right)  +\rho F_{i}(t)\\
&  =\rho\ddot{a}_{i}(t)+C_{i}\dot{g}\left(  z\right)  \sum_{j=1}^{m}\ddot
{f}\left(  z\right)  ^{\gamma_{j}}+g(z)\frac{\partial}{\partial x_{i}}%
\sum_{j=1}^{m}f\left(  z\right)  ^{\gamma_{j}}+\delta\rho\left(  \dot
{f}(z)C_{i}+d_{i}(t)\right)  +\rho F_{i}(t)\\
&  =\rho\ddot{a}_{i}(t)+\rho C_{i}\dot{g}(z)\sum_{j=1}^{m}\ddot{f}\left(
z\right)  ^{\gamma_{j}}+g(z)\sum_{j=1}^{m}\gamma_{j}\ddot{f}\left(  z\right)
^{\gamma_{j}-1}\dddot{f}\left(  z\right)  C_{i}\\
&  +\delta\rho\left(  \dot{f}(z)C_{i}+d_{i}(t)\right)  +\rho F_{i}(t)\\
&  =\rho C_{i}\left\{  \xi+\frac{\alpha(N)}{\overset{N}{\underset{i=1}{\Sigma
}}C_{i}^{2}}\dot{g}(z)\sum_{j=1}^{m}\ddot{f}\left(  z\right)  ^{\gamma_{j}%
-1}+\frac{\alpha(N)}{\overset{N}{\underset{i=1}{\Sigma}}C_{i}^{2}}%
g(z)\sum_{j=1}^{m}\gamma_{j}\ddot{f}\left(  z\right)  ^{\gamma_{j}-2}\dddot
{f}\left(  z\right)  +\delta\dot{f}(z)\right\} \\
&  =0,
\end{align}
where we require the following ordinary differential equations:%
\begin{equation}
\left\{
\begin{array}
[c]{c}%
\ddot{a}_{i}(t)=F_{i}(t)+C_{i}\xi+\delta d_{i}(t),\text{ for }i=1,2,...N,\\
a(0)=a_{0},\text{ }\dot{a}(0)=a_{1},\\
\xi+\frac{\alpha(N)}{\overset{N}{\underset{i=1}{\Sigma}}C_{i}^{2}}\dot
{g}(z)\sum_{j=1}^{m}f\left(  z\right)  ^{\gamma_{j}-1}+\frac{\alpha
(N)}{\overset{N}{\underset{i=1}{\Sigma}}C_{i}^{2}}g(z)\sum_{j=1}^{m}\gamma
_{j}f\left(  z\right)  ^{\gamma_{i}-2}\dot{f}\left(  z\right)  +\delta\dot
{f}(z)=0,\\
g(z)>0,\text{ for }z\in(-\infty,\infty).
\end{array}
\right.
\end{equation}

For $\gamma_{1}=1,$ we get,%
\begin{align}
&  \rho\left(  \frac{\partial u_{i}}{\partial t}+\sum_{k=1}^{N}u_{k}%
\frac{\partial u_{i}}{\partial x_{k}}\right)  +\frac{\partial}{\partial x_{i}%
}e^{S}\sum_{j=1}^{m}\rho^{\gamma_{j}}+\delta\rho\frac{\partial}{\partial
x_{i}}\Phi+\rho F_{i}(t)\\
&  =\rho\dot{a}_{it}(t)+\frac{\partial}{\partial x_{i}}\left(  e^{\ln g\left(
z\right)  }\ddot{f}\left(  z\right)  \right)  +\frac{\partial}{\partial x_{i}%
}\left(  e^{\ln g\left(  z\right)  }\sum_{j=2}^{m}\ddot{f}\left(  z\right)
^{\gamma_{j}}\right) \\
&  +\delta\rho\frac{\partial}{\partial x_{i}}\left(  f(z)+\overset{N}%
{\sum_{i=1}}d_{i}(t)x_{i}\right)  +\rho F_{i}(t)\\
&  =\rho\ddot{a}_{i}(t)+C_{i}\dot{g}\left(  z\right)  \ddot{f}\left(
z\right)  +g(z)\frac{\partial}{\partial x_{i}}f\left(  z\right) \\
&  +C_{i}\dot{g}\left(  z\right)  \sum_{j=2}^{m}\ddot{f}\left(  z\right)
^{\gamma_{j}}+g(z)\frac{\partial}{\partial x_{i}}\sum_{j=2}^{m}f\left(
z\right)  ^{\gamma_{j}}+\delta\rho\left(  \dot{f}(z)C_{i}+d_{i}(t)\right)
+\rho F_{i}(t)\nonumber\\
&  =\rho\ddot{a}_{i}(t)+C_{i}\dot{g}(z)\ddot{f}\left(  z\right)
+g(z)\dddot{f}\left(  z\right)  C_{i}+\\
&  \rho C_{i}\dot{g}(z)\sum_{j=2}^{m}\ddot{f}\left(  z\right)  ^{\gamma_{j}%
}+g(z)\sum_{j=2}^{m}\gamma_{j}\ddot{f}\left(  z\right)  ^{\gamma_{j}-1}%
\dddot{f}\left(  z\right)  C_{i}+\delta\rho\left(  \dot{f}(z)C_{i}%
+d_{i}(t)\right)  +\rho F_{i}(t)\\
&  =\rho C_{i}\left\{
\begin{array}
[c]{c}%
\xi+\frac{\alpha(N)}{\overset{N}{\underset{i=1}{\Sigma}}C_{i}^{2}}\dot
{g}(z)+\frac{\alpha(N)}{\overset{N}{\underset{i=1}{\Sigma}}C_{i}^{2}}%
g(z)\frac{\dddot{f}\left(  z\right)  }{f(z)}\\
+\frac{\alpha(N)}{\overset{N}{\underset{i=1}{\Sigma}}C_{i}^{2}}\dot{g}%
(z)\sum_{j=2}^{m}\ddot{f}\left(  z\right)  ^{\gamma_{j}-1}+\frac{\alpha
(N)}{\overset{N}{\underset{i=1}{\Sigma}}C_{i}^{2}}g(z)\sum_{j=2}^{m}\gamma
_{j}\ddot{f}\left(  z\right)  ^{\gamma_{j}-2}\dddot{f}\left(  z\right)
+\delta\dot{f}(z)
\end{array}
\right\} \\
&  =0,
\end{align}
where we require the following ordinary differential equations:%
\begin{equation}
\left\{
\begin{array}
[c]{c}%
\xi+\frac{\alpha(N)}{\overset{N}{\underset{i=1}{\Sigma}}C_{i}^{2}}\dot
{g}(z)+\frac{\alpha(N)}{\overset{N}{\underset{i=1}{\Sigma}}C_{i}^{2}}%
g(z)\frac{\dddot{f}\left(  z\right)  }{f(z)}+\frac{\alpha(N)}{\overset
{N}{\underset{i=1}{\Sigma}}C_{i}^{2}}\dot{g}(z)\sum_{j=2}^{m}\ddot{f}\left(
z\right)  ^{\gamma_{j}-1}\\
+\frac{\alpha(N)}{\overset{N}{\underset{i=1}{\Sigma}}C_{i}^{2}}g(z)\sum
_{j=2}^{m}\gamma_{j}\ddot{f}\left(  z\right)  ^{\gamma_{j}-2}\dddot{f}\left(
z\right)  +\delta\dot{f}(z)=0,\\
g(z)>0\text{, for }z\in(-\infty,\infty).
\end{array}
\right.
\end{equation}
The proof is completed.
\end{proof}

\begin{remark}
The ordinary differential equations (\ref{at}),%
\begin{equation}
\left\{
\begin{array}
[c]{c}%
\ddot{a}_{i}(t)=F_{i}(t)+C_{i}\xi+d_{i}(t),\text{ for }i=1,2,...N,\\
\text{ }a(0)=a_{0},\text{ }\dot{a}(0)=a_{1},
\end{array}
\right.
\end{equation}
are solved by%
\begin{equation}
a_{i}(t)=\int_{0}^{t}\int_{0}^{s}\left(  F_{i}(\eta)+d_{i}(\eta)\right)  d\eta
ds+\frac{C_{i}\xi t^{2}}{2}+a_{1}t+a_{0}.
\end{equation}

\end{remark}

\begin{remark}
In fact, our solutions, (\ref{yy1}) and (\ref{yy2}), can be easily extended to
the systems, with the generalized multiple and nonlinear, damping and Gamma
pressure function:%
\begin{equation}
\left\{
\begin{array}
[c]{c}%
\frac{\partial\rho}{\partial t}+\sum_{k=1}^{N}u_{k}\frac{\partial\rho
}{\partial x_{k}}+\rho\sum_{k=1}^{N}\frac{\partial u_{k}}{\partial x_{k}%
}{\normalsize =}{\normalsize 0,}\\
\rho\left(  \frac{\partial u_{i}}{\partial t}+\sum_{k=1}^{N}u_{k}%
\frac{\partial u_{i}}{\partial x_{k}}\right)  +\sum_{l=1}^{n}\beta_{l}%
\rho\left[  \left(  \sum_{i=1}^{N}u_{i}^{2}\right)  ^{1/2}\right]  ^{p_{l}%
-1}u_{i}+\frac{\partial}{\partial x_{i}}\left(  e^{S}\sum_{j=1}^{m}\lambda
_{i}\rho^{\gamma_{j}}\right) \\
{\normalsize =-\rho}\frac{\delta\partial}{\partial x_{i}}\Phi+\rho
F_{i}(t)\text{, for }i=1,2,...N,\\
S_{t}+\sum_{k=1}^{N}u_{k}\frac{\partial S}{\partial x_{k}}=0,
\end{array}
\right.
\end{equation}
where $\beta_{1},\beta_{2},....,\beta_{n};\lambda_{1},\lambda_{2},....$
,$\lambda_{m};$ and $p_{1},p_{2},...p_{n}\geq0,$ are constants.
\end{remark}

\end{document}